\documentclass[11pt]{article}

\usepackage{cite}
\usepackage{graphicx}
\usepackage{amsfonts}
\usepackage{amssymb}
\usepackage{amsmath}
\usepackage{verbatim}
\usepackage{float}
\usepackage{scrextend}
\usepackage{xspace}
\usepackage{algpseudocode}
\usepackage{tipa}
\usepackage{amsthm}
\usepackage{fullpage}

\usepackage[dvipsnames]{xcolor}
\usepackage[utf8]{inputenc}
\usepackage{amsthm}
\usepackage{indentfirst}
\usepackage{graphicx}
\usepackage{algorithm}
\usepackage{algpseudocode}
\usepackage{algorithmicx}
\usepackage{amssymb}
\usepackage{tipa}
\usepackage{stmaryrd}
\usepackage{color}
\usepackage{cellspace}
\usepackage{colortbl}
\usepackage{caption}
\usepackage{enumitem}
\usepackage{amsmath}
\usepackage{balance}
\usepackage{multirow}

\newtheorem{theorem}{Theorem}[section]

\newtheorem{lemma}{Lemma}[section]

\floatstyle{ruled}
\newfloat{algorithm}{tp}{lop}
\floatname{algorithm}{Algorithm}

\newcommand{\vs}{{\em vs.}\xspace}
\newcommand{\ie}{{\em i.e.,}\xspace}
\newcommand{\eg}{{\em e.g.,}\xspace}
\newcommand{\etc}{{\em etc.}}
\newcommand{\aka}{{\em a.k.a.}}

\newcommand{\PEFRRR}{$\mathbb{PEF}\_3+$}
\newcommand{\PEFRR}{$\mathbb{PEF}\_2$}
\newcommand{\PEFR}{$\mathbb{PEF}\_1$}

\title{Computability of Perpetual Exploration\\ in Highly Dynamic Rings}

\author{\IEEEauthorblockN{Marjorie Bournat\IEEEauthorrefmark{1}, Swan Dubois\IEEEauthorrefmark{1}, and Franck Petit\IEEEauthorrefmark{1}}
\IEEEauthorblockA{\IEEEauthorrefmark{1}UPMC Sorbonne Universit\'es, CNRS, Inria, LIP6 UMR 7606, France}}

\title{Computability of Perpetual Exploration\\ in Highly Dynamic Rings\thanks{This 
work was performed within Project ESTATE (Ref. ANR-16-CE25-0009-03), supported by 
French state funds managed by the ANR (Agence Nationale de la Recherche).}}

\author{
Marjorie Bournat\footnotemark[2]
\and
Swan Dubois\footnotemark[2]
\and
Franck Petit\footnotemark[2]
}

\footnotetext[2]{UPMC Sorbonne Universit\'es, CNRS, Inria, LIP6 UMR 7606, France}

\date{}

\begin{document}

\maketitle

   \begin{abstract}

We consider systems made of autonomous mobile robots evolving in highly dynamic discrete environment
\ie graphs where edges may appear and disappear unpredictably without any recurrence, stability,  nor periodicity assumption. 
Robots are uniform (they execute the same algorithm), they are anonymous (they are devoid of any observable ID), 
they have no means allowing them to communicate together,
they share no common sense of direction, and they have no global knowledge related to the size of the environment. 
However, each of them is endowed with persistent memory and is able to detect whether it stands alone at its current location. 
A highly dynamic environment is modeled by a graph such that its topology keeps continuously changing over time.   
In this paper, we consider only dynamic graphs in which nodes are anonymous, each of them is infinitely often reachable from 
any other one, and such that its underlying graph (\ie
the static graph made of the same set of nodes and that includes all edges that are present at least once over time) 
forms a ring of arbitrary size.  

In this context, we consider the fundamental problem of perpetual exploration: each node is required to be 
infinitely often visited by a robot. 
This paper 
analyzes the computability of this problem in (fully) synchronous settings, \ie we study 
the deterministic solvability of the problem with respect to the number of robots.
We provide three algorithms and two impossibility results that characterize, for any ring size, 
the necessary and sufficient number of robots to perform perpetual exploration of highly dynamic rings.
%


\end{abstract}

   \newpage
   \section{Introduction}

We consider systems made of autonomous robots that are endowed with visibility sensors and motion actuators.
Those robots must collaborate to perform collective tasks, typically, environmental monitoring, large-scale 
construction, mapping, urban search and rescue, surface cleaning, risky area surrounding, patrolling, exploration 
of unknown environments, to quote only a few.  

Exploration belongs to the set of basic task components for many of the aforementioned applications. For instance, 
environmental monitoring, patrolling, search and rescue, and surface cleaning are all tasks requiring that robots 
(collectively) explore the whole area. To specify how the exploration is achieved, the so-called ``area'' is often 
considered as ``zoned area'' (\eg a building, a town, a factory, a mine, \etc) modeled by a finite graph where 
(anonymous) nodes represent locations that can be sensed by the robots, and edges represent the possibility for a 
robot to move from one location to the other. 

To fit various applications and environments, numerous variants of exploration have been studied in the literature, 
for instance, {\em terminating} exploration ---the robots stop moving after completion of the exploration of the whole graph~
\cite{FIPS07,DPT13,DLPT15}---, {\em exclusive perpetual} exploration ---every node is visited infinitely often, but
no two robots collide at the same node~\cite{BBMR08,BMPT10}---, {\em exploration with return} ---each robot comes 
back to its initial location once the exploration is completed~\cite{DFKP04}\break ---, \etc. Clearly, some of these variants may 
be mixed (\eg exclusive perpetual exploration \vs non exclusive terminating exploration) and either weakened or 
strengthened ({\em weak} perpetual exploration ---every node is visited infinitely often by at least one robot~
\cite{BDPT14}--- \vs~{\em strong} perpetual exploration ---every node is visited infinitely often by each robot---, \etc).
Note that all these instances of exploration are different problems in the sense that, in most of the cases, 
solutions for any given instance cannot be used to solve another instance. Also, some solutions are designed for 
specific graph topologies, \eg ring-shaped~\cite{FIPS07}, line-shaped~\cite{FIPS11}, tree-shaped~\cite{FIPS10}, and 
other for arbitrary network~\cite{CFMS10}. In this paper, we address the (non-exclusive weak version of the) 
{\em perpetual exploration} problem, \ie each node is visited infinitely often by a robot.

Robots operate in \emph{cycles} that include three phases: \emph{Look}, \emph{Compute}, and \emph{Move} (L-C-M).  The
Look phase consists in taking a snapshot of the (local) environment of robots using the visions capabilities offered by the 
sensors they are equipped with. The snapshot depends on the sensor capabilities with respect to environment.  
During the Compute phase, a robot computes a destination based on the previous observation. The Move phase simply 
consists in moving to this destination. Using L-C-M cycles, several models has been proposed in the literature, 
capturing various degrees of synchrony between robots~\cite{FPS12}. They are denoted by $\mathcal{FSYNC}$, 
$\mathcal{SSYNC}$, and $\mathcal{ASYNC}$, from the stronger to the weaker. In $\mathcal{FSYNC}$ ({\em fully 
synchronous}), all robots execute the L-C-M cycle synchronously and atomically.  In $\mathcal{SSYNC}$ ({\em semi-
synchronous}), robots are asynchronously activated to perform cycles, yet at each activation, a robot executes 
one cycle atomically. In $\mathcal{ASYNC}$ ({\em asynchronous}), robots execute L-C-M in a fully independent manner.

We assume robots having weak capabilities: they are \emph{uniform} ---meaning that all robots follow the same algorithm---,
they are \emph{anonymous} ---meaning that no robot can distinguish any two other robots---, they 
are \emph{disoriented} ---they have no coherent labeling of direction---, and they have no global knowledge related to the size 
of the environment. Furthermore, the robots have no (direct) means of communicating with each other. However, 
each of them is endowed with persistent memory and is able to detect whether it stands alone at its current location.

All the aforementioned contributions assume a {\em static} environment, \ie the graph topology explored by the robots 
does not evolve in function of the time. In this paper, we consider {\em dynamic} environments that may change over time, 
for instance, a transportation network, a building in which doors are closed and open over time, or streets that are 
closed over time due to work in process or traffic jam in a town. More precisely, we consider dynamic graphs in which 
edges may appear and disappear unpredictably without any stability, recurrence, nor periodicity assumption. However, 
to ensure that the problem is not trivially unsolvable, we made the assumption that each node is infinitely often 
reachable from any other one through a {\em temporal path} (\aka~{\em journey}~\cite{CFQS12}). The dynamic graphs satisfying this 
topological property are known as {\em connected-over-time} (dynamic) graphs~\cite{CFQS12}.   

{\em Related work.} Recent work~\cite{FMS09,IW11,IW13,IKW14,DDFS16} deal with the terminating exploration of dynamic 
graphs. This line of work restricts the dynamicity of the graph with various assumptions. In~\cite{FMS09} and~\cite{IW11}, the 
authors focus on periodically varying graphs, \ie the presence of each edge of the graph is periodic. In~
\cite{IW13,IKW14,DDFS16}, the authors assume that the graph is connected at each time instant and that there exists a stability 
of this connectivity in any interval of time of length $T$ (such assumption 
is known as $T$-interval-connectivity~\cite{KLO10}). In~\cite{IW13} and~\cite{DDFS16} (resp.~\cite{IKW14}), the authors restrict 
their study to the case where the underlying graph (\ie the static graph that includes all edges that are present 
at least once in the lifetime of the graph) forms a ring (resp. a cactus) of arbitrary size.

In~\cite{DDFS16}, the authors examine the impact of various factors (\eg at least one node is not anonymous, knowledge of the exact number of nodes, 
knowledge of an upper bound on the number of nodes, sharing of a common orientation, \etc) on the solvability of the terminating
exploration. In particular, they show that the degree of synchrony among the robot has a major impact. Indeed, they 
prove that, independently of other assumptions, exploration is impossible in $\mathcal{SSYNC}$ model (without extra 
synchronization assumptions). The proof of this result relies on the possibility offered to the adversary to
wake up each robot independently and to remove the edge that the robot wants to traverse at this time. Note that, 
by its simplicity, this impossibility result is applicable to any variante of the exploration problem. It is also 
independent of dynamicity assumptions. 

The first attempt to solve exploration in the most general dynamicity scenario (\ie connected-over-time assumption)
has been proposed in~\cite{BDD16}. The authors provide a protocol that deterministically solves the perpetual 
exploration problem. This protocol operates in any connected-over-time ring with three synchronous robots 
(accordingly to the aforementioned impossibility result in~\cite{DDFS16}). Further, the proposed protocol has the 
nice extra property of being self-stabilizing, meaning that regardless their arbitrary initial configuration, the 
robots eventually behave according to their specification, \ie eventually, they explore the whole network infinitely 
often. Note that the necessity of the assumption on the number of robots is left as an open question by this work. 
 


{\em Our contribution.} The main contribution of this paper is to close this question. Indeed, we analyze the 
computability of the perpetual exploration problem in connected-over-time (dynamic) rings, 
\ie we study the deterministic solvability of the problem with respect to the number of robots. According to the 
impossibility result in~\cite{DDFS16}, we restrict this study to the $\mathcal{FSYNC}$ model. As we do not consider 
self-stabilization (contrarily to~\cite{BDD16}), we assume that no pair of robots have a common initial 
location. Moreover, to ensure that the problem is not trivially solved in the initial configuration, we consider
that, $k$, the number of robots, is strictly smaller 
than $n$, the number of nodes of the dynamic graph.
In this context, we establish the necessary and sufficient number of robots to solve the perpetual exploration 
for any size of connected-over-time rings (see TABLE~\ref{table} for a summary). 
Note that a connected-over-time chain can be seen as a connected-over-time ring with a missing edge.  
So, our results are also valid on connected-over-time chains.

\begin{table}
\begin{center}
  \begin{tabular}{| c | c| c |}
      \hline
	 Number of Robots &Size of Rings & Results \\ \hline \hline
	
	 {3 and more} & $\geq$ 4 & Possible (Theorem~\ref{final_theorem})  \\ \hline
	
	\multirow{2}*{2} &  $>$ 3 & Impossible (Theorem~\ref{no_perpetual_exploration_two_robots})  \\
	\cline{2-3}  & = 3 & Possible (Theorem~\ref{th:algo2robots}) \\ \hline
	
	 \multirow{2}*{1} & $>$ 2 & Impossible (Theorem~\ref{no_perpetual_exploration_one_robot})    \\
	\cline{2-3}  & = 2 & Possible (Theorem~\ref{th:algo1robot}) \\ \hline
	    
  \end{tabular}
  \caption{Overview of the results} \label{table}
  \end{center}
\end{table}

In more details, we first provide an algorithm that perpetually explores, using a team of $k\geq 3$ robots, any 
connected-over-time ring of $n > k$ nodes. Then, we give two non-trivial impossibility results. 
We first show that two robots are not sufficient to perpetually explore a 
connected-over-time ring with a number of nodes strictly greater than three. Next, we show that
a single robot cannot perpetually explore a connected-over-time ring with a number of 
nodes strictly greater than two. Finally, we close the problem 
by providing an algorithm for each remaining cases (one robot in a $2$-node connected-over-time ring and
two robots in a $3$-node connected-over-time ring). 
%


%
%

{\em Outline of the paper.} In Section \ref{sec:model}, we present formally the model considered in the
remainder of the paper. Section \ref{sec:3robots} presents the algorithm to explore connected-over-time rings of size $n>k$ nodes
with $k\geq 3$ robots. The impossibility result and the algorithm for two robots are both presented in Section 
\ref{sec:2robots}.   The ones assuming a single robot are given in Section \ref{sec:1robot}. We conclude 
in Section \ref{sec:conclu}.

   \section{Model}\label{sec:model}

In this section, we present our formal model. This model is borrowed from the 
one of~\cite{BDD16} that proposes an extension of the classical model of robot networks 
in static graphs introduced in~\cite{KMP06} to the context of dynamic graphs.
 
\subsection{Dynamic graphs} In this paper, we consider the model of 
\emph{evolving graphs} introduced in \cite{XFJ03}. We hence consider the time as
discretized and mapped to $\mathbb{N}$. An evolving graph $\mathcal{G}$ is an
ordered sequence $\{G_{0}, G_{1}, \ldots\}$ of subgraphs of a given static graph 
$G=(V,E)$ such that, for any $i\geq 0$, we have $G_{i} = (V, E_{i})$. We say that
the edges of $E_{i}$ are \emph{present} in $\mathcal{G}$ at time $i$. 
The \emph{underlying graph} of $\mathcal{G}$, denoted $U_\mathcal{G}$, is the 
static graph gathering all edges that are present at least once in $\mathcal{G}$ 
(\ie $U_\mathcal{G}=(V,E_\mathcal{G})$ with
$E_\mathcal{G}=\bigcup_{i=0}^{\infty}E_i)$).
An \emph{eventual missing edge} is an edge of $E_\mathcal{G}$ such that there 
exists a time after which this edge is never present in $\mathcal{G}$. A 
\emph{recurrent edge} is an edge of $E_\mathcal{G}$ that is not eventually 
missing. The \emph{eventual underlying graph} of $\mathcal{G}$, denoted
$U_\mathcal{G}^\omega$, is the static graph gathering all recurrent edges of
$\mathcal{G}$ (\ie $U_\mathcal{G}^\omega=(V,E_\mathcal{G}^\omega)$ where 
$E_\mathcal{G}^\omega$ is the set of recurrent edges of $\mathcal{G}$).

In this paper, we chose to make minimal assumptions on the dynamicity of
our graph since we restrict ourselves on \emph{connected-over-time} evolving 
graphs. The only constraint we impose on evolving graphs of this class is that 
their eventual underlying graph is connected \cite{DKP15} (this is equivalent with
the assumption that each node is infinitely often reachable from another one
through a journey). In the following, we consider only connected-over-time 
evolving graphs whose underlying graph is an anonymous and unoriented ring of 
arbitrary size. Although the ring is unoriented, to simplify the presentation,
we, as external observers, distinguish between
the clockwise and the counter-clockwise (global) direction in the ring.

We introduce here some definitions that are used for proofs only.
From an evolving graph $\mathcal{G}=\{(V, E_0), (V, E_1),(V, E_2),$ $\ldots\}$, 
we define the evolving graph 
$\mathcal{G} \backslash \{(e_{1}, \tau_{1}), \ldots (e_{k}, \tau_{k})\}$ (with 
for any $i \in \{1, \ldots, k\}$, $e_{i} \in E$ and $\tau_{i} \subseteq \mathbb{N}$) as the 
evolving graph $\{(V, E_0'), (V, E_1'), (V, E_2'),\ldots\}$ such that:
$\forall t\in\mathbb{N},\forall e\in E_\mathcal{G},
e\in E_t' \Leftrightarrow e\in E_t\wedge(\forall i\in\{1,\ldots,k\}, e\neq e_i \vee t\notin \tau_i)$.
A node $u$ satisfies the property $OneEdge(u, t, t')$ if and only if an 
adjacent edge of $u$ is continuously missing from time $t$ to time $t'$ while 
the other adjacent edge of $u$ is continuously present from time $t$ to time 
$t'$. We define the distance between two nodes $u$ and $v$ (denoted $d(u,v)$)
by the length of a shortest path between $u$ and $v$ in the underlying graph.

\subsection{Robots} We consider systems of autonomous mobile entities called
robots moving in a discrete and dynamic environment modeled by an evolving graph
$\mathcal{G}=\{(V,E_0),(V,E_1)\ldots\}$, $V$ being a set of nodes representing 
the set of locations where robots may be, $E_i$ being the set of bidirectional 
edges representing connections through which robots may move from a location to
another one at time $i$. Robots are uniform (they execute the same algorithm), 
anonymous (they are indistinguishable from each other), and have a persistent 
memory (they can store local variables). The state of a robot at time $t$
corresponds to the value of its variables at time $t$. Robots are unable to 
directly communicate with each other by any means. Robots are endowed with 
local weak multiplicity detection meaning that they are able to detect if they
are alone on their current node or not, but they cannot know the exact number of 
co-located robots. When a robot is alone on its current node, we say that it is 
isolated. A \emph{tower} $T$ is a couple $(S, \theta)$, where $S$ 
is a set of robots ($|S| > 1$) and $\theta=[t_{s}, t_{e}]$ is an interval of
$\mathbb{N}$, such that all the robots of $S$ are located at a same node at each
instant of time $t$ in $\theta$ and $S$ or $\theta$ is maximal for this
property. We say that the robots of $S$ form the tower at time $t_{s}$ and that 
they are involved in the tower between time $t_{s}$ and $t_{e}$. Robots have no a
priori knowledge about the ring they explore (size, diameter, dynamicity\ldots).
Finally, each robot has its own stable chirality (\ie each robot is able to
locally label the two ports of its current node with \emph{left} and
\emph{right} consistently over the ring and time but two different robots may
not agree on this labeling). We assume that each robot has a variable $dir$ that
stores a direction (either \emph{left} or \emph{right}). Initially, this variable 
is set to $left$. At any time, we say that a robot points to \emph{left} 
(resp. \emph{right}) if its $dir$ variable is equal to this (local) direction.
Through misuse of language, we say that a robot points to an edge when 
this edge is connected to the current node of the robot by the port labeled
with its current direction. 
We say that a robot considers the clockwise (resp. counter-clockwise) direction
if the (local) direction pointed to by this robot corresponds to the (global)
direction seen by an external observer. 

\subsection{Execution} A configuration $\gamma$ of the system captures the 
position (\ie the node where the robot is currently located) and the state of 
each robot at a given time. Given an evolving graph 
$\mathcal{G}=\{G_{0}, G_{1}, \ldots\}$, an algorithm $\mathcal{A}$, and an
initial configuration $\gamma_0$, the execution $\mathcal{E}$ of $\mathcal{A}$
on $\mathcal{G}$ starting from $\gamma_0$ is the infinite sequence 
$(G_0,\gamma_0),(G_1,\gamma_1),(G_2,\gamma_2),\ldots$ where, for any $i\geq 0$, 
the configuration $\gamma_{i+1}$ is the result of the execution of a synchronous
round by all robots from $(G_i,\gamma_i)$ as explained below.

The round that transitions the system from $(G_i,\gamma_i)$ to
$(G_{i+1},\gamma_{i+1})$ is composed of three atomic and synchronous phases: 
Look, Compute, Move. During the Look phase, each robot gathers information about
its environment in $G_i$. More precisely, each robot updates the value of the 
following local predicates:
$(i)$ $ExistsEdge(dir)$ returns true if there is an adjacent edge at the 
current location of the robot on its direction $dir$, false otherwise; 
$(ii)$ $ExistsOtherRobotsOnCurrentNode()$ returns true if there is strictly more than  
one robot on the current node of the robot, false otherwise. 
We define the local environment of a robot at a given time as the combination of the 
values of $ExistsEdge(dir)$, $ExistsEdge(\overline{dir})$ (where $\overline{dir}$ 
is the opposite direction to $dir$), and $ExistsOtherRobotsOnCurrentNode()$ of 
this robot at this time. The view of the robot at this time gathers its state 
and its local environment at this time. During the 
Compute phase, each robot executes the algorithm $\mathcal{A}$ that may modify 
its variable $dir$ depending on its current state and on the values of the 
predicates updated during the Look phase. Finally, the Move phase consists of
moving each robot through one edge in the direction it points to if there exists
an edge in that direction, otherwise (\ie the edge is missing at that time)
the robot remains at its current node.

\subsection{Specification} We define a well-initiated execution as an execution
$(G_0,\gamma_0),(G_1,\gamma_1),(G_2,\gamma_2),\ldots$ such that $\gamma_0$ contains 
strictly less robots than the number of nodes of $\mathcal{G}$ and is towerless 
(\ie there is no tower in this configuration).

Given a class of evolving graphs $\mathcal{C}$, an algorithm $\mathcal{A}$
satisfies the perpetual exploration specification on $\mathcal{C}$ if and only 
if, in every well-initiated execution
of $\mathcal{A}$ on every evolving graph $\mathcal{G}\in\mathcal{C}$, every node of
$\mathcal{G}$ is infinitely often visited by at least one robot (\ie a robot is
infinitely often located at every node of $\mathcal{G}$). Note that this 
specification does not require that every robot visits infinitely often every 
node of $\mathcal{G}$.



   \section{With Three or More Robots}\label{sec:3robots}

This section is dedicated to the more general result: the perpetual exploration 
exploration of connected-over-time rings of size greater than $k$
with a team of $k \geq 3$ robots.

\subsection{Presentation of the Algorithm}

We first describe intuitively the key ideas of our algorithm. Remind that an 
algorithm controls the move of the robots through their variable direction. Hence,
designing an algorithm consists in choosing when we want a robot to keep its 
direction and when we want it to change its direction (in other words, turn back). 
The first idea of our algorithm is to require that a robot keeps its direction 
when it is not involved in a tower (Rule $1$). Using this idea, some towers are 
necessarily formed when there exists an eventual 
missing edge. Our algorithm reacts as follows to the formation of towers.
If at a time $t$ a robot does not move and forms a tower at
time $t + 1$, then the algorithm keeps the direction of the robot (Rule $2$).
In the contrary case (that is, at time $t$, the  robot moves and forms a tower 
at time t+1) it changes the direction of the robot (Rule $3$).

Let us now explain how the algorithm (Rules $1$, $2$, and $3$) enables the 
perpetual exploration of any connected-over-time ring. First, note that Rule $1$
alone is sufficient to perpetually explore connected-over-time rings without 
eventual missing edge provided that the robots never meet. The main property 
induced by Rules $2$ and $3$ is that any tower is broken in a finite time
and that at least one robot of the tower considers each possible
direction. This property implies (combined with Rule $1$) 
that $(i)$ the algorithm is able to perpetually explore
any connected-over-time ring without eventual missing edge (even if robots meet); 
and that $(ii)$, when the ring contains an eventual missing edge, one robot is 
eventually located at each extremity of the eventual missing edge and considers 
afterwards the direction of the eventual missing edge. 

Let us consider this last case. We call sentinels the two robots located at extremities 
of the eventual missing edge. The other robots are called explorers. 
By Rule $3$, an explorer that arrives on a node where a sentinel is located
changes its direction. Intuitively, that means that the sentinel signal 
to the explorer that it has reached one extremity of the eventual missing edge and
that it has consequently to turn back to continue the exploration.
Note that, by Rule $2$, the sentinel keeps its direction (and hence its role).
Once an explorer leaves an extremity of the eventual missing edge, we know, 
thanks to Rule $1$ and the main property induced by Rules $2$ and $3$, 
that a robot reaches in a finite time the other extremity of the eventual missing edge
and that (after the second sentinel/explorer meeting) all the nodes have been 
visited by a robot in the meantime. As we can repeat this scheme
infinitely often, our algorithm is able to perpetually explore any connected-over-time 
ring with an eventual missing edge, that ends the informal presentation of our algorithm.


Refer to Algorithm \ref{algo:pef} for the formal statement of our algorithm
called \PEFRRR~(standing for $\mathbb{P}$erpetual $\mathbb{E}$xploration 
in $\mathbb{F}$SYNC with 3 or more robots).
In addition of its $dir$ variable, each robot maintains a boolean variable
$HasMovedPreviousStep$ indicating if the robot has moved during its last 
Look-Compute-Move cycle. This variable is used to implement Rules $2$ and 
$3$.

\begin{algorithm}
\caption{\PEFRRR}\label{algo:pef}
	\begin{algorithmic} [1]
		\If {$HasMovedPreviousStep$ $\wedge$ $ExistsOtherRobotsOnCurrentNode()$}
			\State $dir \leftarrow \overline{dir}$
		\EndIf
		\State $HasMovedPreviousStep \leftarrow ExistsEdge(dir)$
	\end{algorithmic}
\end{algorithm}

\subsection{Proof of Correctness}

In this section, we prove the correctness of \PEFRRR~with $k\geq 3$ robots.
In the following, we consider a connected-over-time ring $\mathcal{G}$ of  
size at least $k+1$. Let $\varepsilon = (G_0,\gamma_0),(G_1,\gamma_1),\ldots$
 be any execution of \PEFRRR~on $\mathcal{G}$.


\begin{lemma} \label{eventualMissingEdgeMeetingHappens}
 If there exists an eventual missing edge in $\mathcal{G}$, then at least one 
 tower is formed in $\varepsilon$.
\end{lemma}

\begin{proof}
 By contradiction, assume that $e$ is an eventual missing edge of $\mathcal{G}$ 
 (such that $e$ is not present in $\mathcal{G}$ after time $t$) and that
 no tower is formed in $\varepsilon$. 

 Executing \PEFRRR, a robot changes the global direction it considers only when
 it forms a tower with another robot. As, by assumption, no tower is formed in
 $\varepsilon$, each robot is always considering the same global direction. All
 the edges of $\mathcal{G}$, except $e$, are infinitely often present in 
 $\mathcal{G}$. Hence, any robot reaches one of the extremity of $e$ in finite
 time after $t$. As the robots consider a direction at each instant time and 
 that there are at least 3 robots, at least 2 robots consider the same global
 direction at each instant time. Hence, at least two robots reach the same
 extremity of $e$. A tower is formed, leading to a contradiction.
\end{proof}

\begin{lemma} \label{noMeetingPerpetualExplorationSolved}
 If $\varepsilon$ does not contain a tower, then every node is infinitely often 
 visited by a robot in $\varepsilon$.
\end{lemma}

\begin{proof}
 Assume that there is no tower formed in $\varepsilon$. 
 By Lemma~\ref{eventualMissingEdgeMeetingHappens}, if there is an eventual 
 missing edge in $\mathcal{G}$, then there is at least one tower formed. In 
 consequence, all the edges of $\mathcal{G}$ are infinitely often present in 
 $\mathcal{G}$. 
 
 Executing \PEFRRR, a robot changes the global direction it considers only when
 it forms a tower with another robot. Hence, none of the robots change the
 global direction it considers in $\varepsilon$. Since all the edges are
 infinitely often present, each robot moves infinitely often in the same global
 direction, that implies the result.
\end{proof}

\begin{lemma} \label{oppositeDirectionAfterTower}
 If a tower $T$ of 2 robots is formed in $\varepsilon$, then these two robots
 consider two opposite global directions while $T$ exists.
\end{lemma}

\begin{proof}
 Assume that 2 robots form a tower at a time $t$ in $\varepsilon$. Let us 
 consider the 2 following cases:
 
  \noindent\textbf{Case 1:} The two robots consider the same global direction 
  during the Move phase of time $t - 1$. 

  \noindent In this case, one robot (denoted $r$) does not move during the
  Move phase of time $t$, while the other (denoted $r'$) moves and joins the 
  first one on its current node. During the Compute phase of time $t$, $r$ still
  considers the same global direction, while $r'$ changes the global direction
  it considers by construction of \PEFRRR. Then, the two robots consider two 
  different global directions after the Compute phase of time $t$.
  
  \noindent\textbf{Case 2:} The two robots consider two opposite global directions 
  during the Move phase of time $t - 1$. 

  \noindent In this case, the two robots move at
  time $t - 1$. During the Compute phase of time $t$, the two robots change the
  global direction they consider by construction of \PEFRRR. Hence they consider 
  two different global directions after the Compute phase of time $t$.

 A robot executing \PEFRRR~changes its global direction only if it
 has moved during the previous step. So, the robots of the tower do not change
 the global direction they consider as long as they are involved in the tower. As 
 the two robots consider two different global directions after the Compute phase
 of time $t$, we obtain the lemma.
\end{proof}

\begin{lemma} \label{NoMoreThan2RobotsMeeting}
 No tower of $\varepsilon$ involves 3 robots or more.
\end{lemma}

\begin{proof}
 We prove this lemma by recurrence. As there is no tower in $\gamma_{0}$ by 
 assumption, it remains to prove that, if $\gamma_{t}$ contains no tower
 with 3 or more robots, so is $\gamma_{t + 1}$. Let us study the following 
 cases:
 
  \noindent\textbf{Case 1:} $\gamma_{t}$ contains no tower.
 
  \noindent The robots can cross at most one edge at each step. Each node has at most 2 
  adjacent edges in $G_{t}$, hence the maximum number of robots involved in a 
  tower of $\gamma_{t + 1}$ is 3. If a tower involving 3 robots is formed in 
  $\gamma_{t + 1}$, one robot $r$ has not moved during the Move phase of time 
  $t$, while the two other robots (located on the two adjacent nodes of its 
  location) have moved to its position. That implies that the two adjacent edges
  of the node where $r$ is located are present in $G_{t}$. As any robot 
  considers a global direction at each instant time, $r$ necessarily moves in
  step $t$, that is contradictory. Therefore, only towers of 2 robots can be 
  formed in $\gamma_{t + 1}$.
  
  \noindent\textbf{Case 2:} $\gamma_{t}$ contains towers of at most 2
  robots.

  \noindent Let $T$ be a tower involving 2 robots in $\gamma_{t}$ and $u$ be the node
  where $T$ is located in $\gamma_{t}$. By 
  Lemma~\ref{oppositeDirectionAfterTower}, the 2 robots of $T$ consider two 
  opposite global directions in $\gamma_{t}$. 
  
  \noindent Consider the 3 following sub-cases:
  
  \noindent $(i)$ If there is no adjacent edge to $u$ in $G_{t}$, then no other robot can
  increase the number of the robots involved in the tower. 
  
  \noindent $(ii)$ If there is only one adjacent edge to $u$ in $G_{t}$, then only one 
  robot may traverse this edge to increase the number of robots involved in $T$.
  Indeed, if there are multiple robots on an adjacent node to $u$, then these 
  robots are involved in a tower $T'$ of 2 robots (by assumption on 
  $\gamma_{t}$) and they are considering two opposite global directions in 
  $\gamma_{t}$. However, as an adjacent edge to $u$ is present in 
  $G_{t}$ and as the robots of $T$ are considering two opposite global 
  directions, then one robot of $T$ leaves $T$ at time $t$. In other words, even
  if a robot of $T'$ moves on $u$, one robot of $T$ leaves $u$. Then, there is 
  at most 2 robots on $u$ in $\gamma_{t + 1}$. 
  
  \noindent $(iii)$ If there are two adjacent edges to $u$ in $\gamma_{t}$, then, using 
  similar arguments as above, we can prove that only one robot crosses each
  of the adjacent edges of $u$. Moreover, the robots of $T$ move in opposite 
  global directions and leave $u$, implying that at most 2 robots are 
  present on $u$ in $\gamma_{t + 1}$.
\end{proof}

\begin{lemma} \label{NormalCasePerpetualExplorationSolved}
 If $\mathcal{G}$ has no eventual missing edge and $\varepsilon$ contains towers
 then every node is infinitely often visited by a robot in $\varepsilon$.
\end{lemma}

\begin{proof}
 Assume that $\mathcal{G}$ has no eventual missing edge and $\varepsilon$ 
 contains towers.
 
 
 We want to prove the following property. If during the Look phase
 of time $t$, a robot $r$ is located on a node $u$ considering the global direction $gd$, 
 then there exists a time $t' \geq t$ such that, during the Look phase of time 
 $t'$, a robot is located on the node $v$ adjacent to $u$ in the global direction $gd$ 
 and considers the global direction $gd$.
 
 Let $t" \geq t$ be the smallest time after time $t$ where the 
 adjacent edge of $u$ in the global direction $gd$ is present in $\mathcal{G}$.
 As all the edges of $\mathcal{G}$ are infinitely often present, $t"$ exists.
   
  \noindent $(i)$ If $r$ crosses the adjacent edge of $u$ in the global direction $gd$ 
  during the Move phase of time $t"$, then the property is verified.
  
  \noindent $(ii)$ If $r$ does not cross the adjacent edge of $u$ in the global direction
  $gd$, this implies that $r$ changes the global direction it considers during the
  Look phase of time $t$. While executing \PEFRRR, a robot changes its global 
  direction when it forms a tower with another robot. Therefore, at 
  time $t$, $r$ forms a tower with a robot $r'$. By 
  Lemmas~\ref{NoMoreThan2RobotsMeeting} and \ref{oppositeDirectionAfterTower}, 
  two robots involved in a tower consider two opposite global directions. Hence,
  after the Compute phase of time $t$, $r'$ considers the global direction $gd$.
  A robot executing \PEFRRR~does not change the global direction it considers
  until it moves. So, $r'$ considers the global direction $gd$ during 
  the Move phase of time $t"$. Hence, during the Look phase of time $t" + 1$, 
  $r'$ is on node $v$ and considers the global direction $gd$.
 
 By applying recurrently this property to any robot, we
 prove that all the nodes are infinitely often visited.
\end{proof}

\begin{lemma} \label{ReachabilityOfAdjacentNodesOfAMissingEdge}
 If $\mathcal{G}$ has an eventual missing edge $e$ (such that $e$ is missing 
 forever after time $t$) and, during the Look phase of a time $t' \geq t$, a
 robot considers a global direction $gd$ and is located on a node at a distance
 $d \neq 0$ in $U_\mathcal{G}^\omega$ from the extremity of $e$ in the global
 direction $gd$, then it exists a time $t" \geq t'$ such that, during the Look 
 phase of time $t"$, a robot is on a node at distance $d - 1$ in
 $U_\mathcal{G}^\omega$ from the extremity of $e$ in the global direction $gd$
 and considers the global direction $gd$.
\end{lemma}
 
\begin{proof}
 Assume that $\mathcal{G}$ has an eventual missing edge $e$ (such that $e$ is 
 missing forever after time $t$) and that, during the Look phase of time 
 $t' \geq t$, a robot $r$ considers a global
 direction $gd$ and is located on a node $u$ at distance $d \neq 0$ in 
 $U_\mathcal{G}^\omega$ from the extremity of $e$ in the global direction $gd$. 
 
 Let $v$ be the adjacent node of $u$ in the global direction $gd$.

 Let $t" \geq t'$ be the smallest time after time $t'$ where the 
 adjacent edge of $u$ in the global direction $gd$ is present in $\mathcal{G}$.
 As all the edges of $\mathcal{G}$ except $e$ are infinitely often present and
 as $u$ is at a distance $d \neq 0$ in $U_\mathcal{G}^\omega$ from the extremity
 of $e$ in the global direction $gd$, then the adjacent edge of $u$ in the 
 global direction $gd$ is infinitely often present in $\mathcal{G}$. Hence, $t"$ 
 exists.
   
  \noindent $(i)$ If $r$ crosses the adjacent edge of $u$ in the global direction $gd$ 
  during the Move phase of time $t"$, then the property is verified.
  
  \noindent $(ii)$ If $r$ does not cross the adjacent edge of $u$ in the global direction
  $gd$, this implies that $r$ changes the global direction it considers during the
  Look phase of time $t$. While executing \PEFRRR~a robot changes the global 
  direction it considers when it forms a tower with another robot. Therefore, at 
  time $t$, $r$ forms a tower with a robot $r'$. By 
  Lemmas~\ref{NoMoreThan2RobotsMeeting} and \ref{oppositeDirectionAfterTower}, 
  two robots involved in a tower consider two opposite global directions. Hence,
  after the Compute phase of time $t$, $r'$ considers the global direction $gd$.
  A robot executing \PEFRRR~does not change the global direction it considers
  until it moves. Therefore, $r'$ considers the global direction $gd$ during 
  the Move phase of time $t"$. Hence, during the Look phase of time $t" + 1$, 
  $r'$ is on node $v$ and considers the global direction $gd$.
\end{proof}
 
\begin{lemma} \label{2RobotsForeverOnAdjacentNodesOfAMissingEdge}
 If $\mathcal{G}$ has an eventual missing edge $e$, then eventually one robot
 is forever located on each extremity of $e$ pointing to $e$.
\end{lemma}
 
\begin{proof}
 Assume that $\mathcal{G}$ has an eventual missing edge $e$ such that $e$ is 
 missing forever after time $t$.
 
 First, we want to prove that a robot reaches one of the extremities of $e$ 
 in a finite time after $t$ and points to $e$ at this time. If it is not 
 the case at time $t$, then there exists at this time a robot
 considering a global direction $gd$ and located on a node $u$ at distance
 $d \neq 0$ in $U_\mathcal{G}^\omega$ from the extremity of $e$ in the global
 direction $gd$. By applying $d$ times
 Lemma~\ref{ReachabilityOfAdjacentNodesOfAMissingEdge}, we 
 prove that, during the Look phase of a time $t'\geq t$, a robot (denote it $r$) 
 reaches the extremity of $e$ in the global
 direction $gd$ from $u$ (denote it $v$ and let $v'$ be the other extremity of $e$), 
 and that this robot considers the global direction
 $gd$. Let us consider the following cases:
 
  \noindent\textbf{Case 1:} $r$ is isolated on $v$ at time $t'$.
  
  \noindent In this case, by construction of \PEFRRR, $r$ does not change,
  during the Compute phase of time $t'$, the global direction that it considers
  during the Move phase of time $t' - 1$. Moreover, a robot can change the
  global direction it considers only if it moves during the previous step.
  All the edges of $\mathcal{G}$ except $e$ are infinitely often
  present. As, at time $t'$, $r$ points to $e$, it cannot move. Therefore, from
  time $t'$, $r$ does not move and does not change the global direction it
  considers. Then, $r$ remains located on $v$ forever after $t'$ considering $gd$.

  \noindent\textbf{Case 2:} $r$ is not isolated on $v$ at time $t'$.
  
  \noindent By Lemmas~\ref{NoMoreThan2RobotsMeeting}, $r$ forms a tower 
  with only one another robot $r'$. By Lemmas~\ref{NoMoreThan2RobotsMeeting} and 
  \ref{oppositeDirectionAfterTower}, two robots that form a tower consider two 
  opposite global directions. Hence, either $r$ or $r'$ 
  considers the global direction $gd$ while the other one consider the global 
  direction $\overline{gd}$. As all the edges of $\mathcal{G}$ except $e$ are 
  infinitely often present, then in finite time either $r$ or $r'$ leaves $v$. 
  We can now apply the same arguments than in Case 1 to the robot that stays 
  on $v$ to prove that this robot remains located on $v$ forever after $t'$ 
  considering $gd$.

 In both cases, a robot remains forever on $v$ considering $gd$ after $t'$. 
 Assume without loss of generality that it is $r$.
 Let us consider the two following cases:

  \noindent\textbf{Case A:} It exists $r'\neq r$ considering $\overline{gd}$ at time $t'$.

  \noindent We can apply recurrently Lemma~\ref{ReachabilityOfAdjacentNodesOfAMissingEdge}, 
  and the arguments above to prove that a robot is eventually forever located on $v'$
  considering $\overline{gd}$.

  \noindent\textbf{Case B:} All robots $r'\neq r$ considers $gd$ at time $t$.

  \noindent We can apply recurrently Lemma~\ref{ReachabilityOfAdjacentNodesOfAMissingEdge}
  to prove that, in finite time, a robot forms a tower with $r$ on $v$. Then,
  by construction of \PEFRRR, this robot consider $\overline{gd}$ after the Compute
  phase of this time (and hence during the Look phase of the next time).
  We then come back to Case A.

 In both cases, the lemma holds.
\end{proof}

\begin{lemma} \label{OneEdgeMissingForeverPerpetualExplorationSolved}
 If $\mathcal{G}$ has an eventual missing edge and $\varepsilon$ contains
 towers, then every node is infinitely often visited.
\end{lemma}

\begin{proof}
 Assume that $\mathcal{G}$ has an eventual missing edge $e$ that is missing 
 forever after time $t$. By 
 Lemma~\ref{2RobotsForeverOnAdjacentNodesOfAMissingEdge}, there exists a time 
 $t' \geq t$ after which two robots $r_{1}$ and $r_{2}$ are respectively located
 on the two extremities of $e$ and pointing to $e$.
 As there are at least 3 robots, let $r$ be a robot (located on a 
 node $u$ considering a global direction $gd$) such that $r \neq r_{1}$ and
 $r \neq r_{2}$. Let $v$ be the extremity of $e$ in the direction $gd$ of $u$
 and $v'$ be the other extremity of $e$.

 Applying recurrently Lemma~\ref{ReachabilityOfAdjacentNodesOfAMissingEdge},
 we prove that, in finite time, all the nodes between node $u$ and $v$ in the 
 global direction $gd$ are visited and that a robot 
 reaches $v$. When this robot reaches $v$, it changes its direction (hence 
 considers $\overline{gd}$) by construction of \PEFRRR~since it moves during the 
 previous step and forms a tower.

 We can then repeat this reasoning (with $v$ and $v'$ alternatively in the role 
 of $u$ and with $v'$ and $v$ alternatively in the role of $v$) and prove that all nodes 
 are infinitely often visited.
\end{proof}

Lemmas~\ref{noMeetingPerpetualExplorationSolved}, 
\ref{NormalCasePerpetualExplorationSolved}, and
\ref{OneEdgeMissingForeverPerpetualExplorationSolved} directly imply the 
following result: 

\begin{theorem} \label{final_theorem}
 \PEFRRR~is a perpetual exploration algorithm for the class of connected-over-time
 rings of arbitrary size strictly greater than the number of robots
 using an arbitrary number (greater than or equal to 3) of fully synchronous robots.
\end{theorem}

   \section{With Two Robots}\label{sec:2robots}

In this section, we study the perpetual exploration of rings of any size with 
two robots. We first prove that two 
robots are not able to perpetually explore connected-over-time rings of size 
strictly greater than three (refer to Theorem~\ref{no_perpetual_exploration_two_robots}).
Then, we provide \PEFRR~(see Theorem \ref{th:algo2robots}),
an algorithm using two robots that solves
the perpetual exploration on the remaining case, \ie connected-over-time rings of size three.

\subsection{Connected-over-Time Rings of Size $4$ or More}\label{sub:imp2robots}

  \begin{figure*}
 	\begin{center}
 		\includegraphics[scale=0.7]{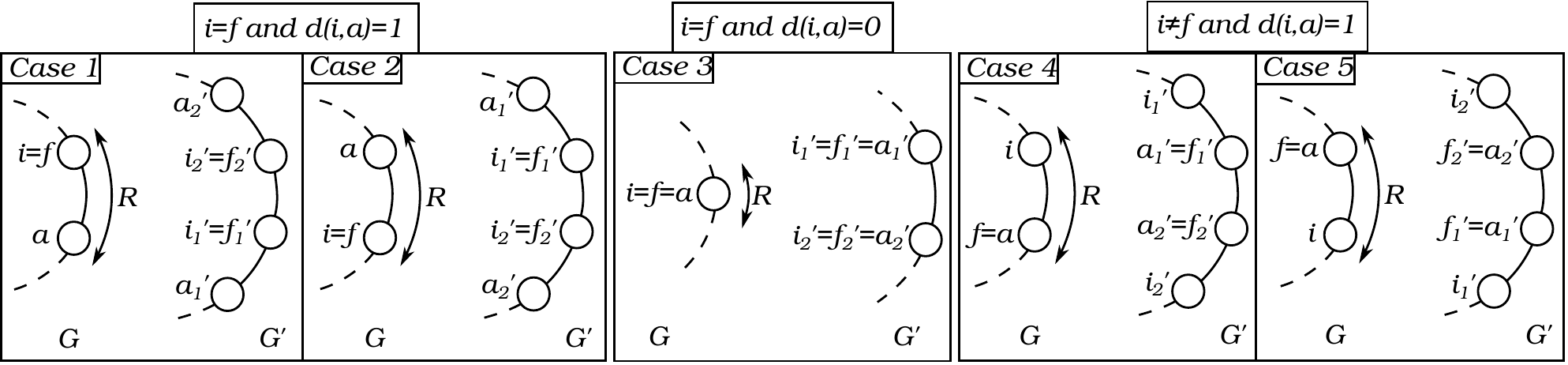}
 	\end{center}
 	\caption{Construction of $\mathcal{G}'$ in proof of Lemma 
            \ref{lemma_modification_direction_bis}.} \label{impossibility_ter}
 \end{figure*}

The proof of our impossibility result presented in Theorem 
\ref{no_perpetual_exploration_two_robots} makes use of a generic framework 
proposed in \cite{BDKP16}. Note that, even if this generic framework is designed 
for another model (namely, the classical message passing model), it is straightforward to
borrow it for our current model. Indeed, its proof only relies on the determinism of
algorithms and indistinguishability of dynamic graphs, these arguments being directly 
translatable in our model. We present briefly this framework here. The interested 
reader is referred to \cite{BDKP16} for more details. 

This framework is based on a theorem that ensures that, if we take a sequence of 
evolving graphs with ever-growing common prefixes (that hence converges to the evolving 
graph that shares all these common prefixes), then the sequence of corresponding 
executions of any deterministic algorithm also converges. 
Moreover, we are able to describe the execution to which 
it converges as the execution of this algorithm on the evolving graph to which the 
sequence converges.
This result is useful since it allows 
us to construct counter-example in the context of impossibility results. Indeed, it is
sufficient to construct an evolving graphs sequence (with ever-growing common prefixes) 
and to prove that their corresponding execution violates the specification of the problem 
for ever-growing time to exhibit an execution that never satisfies the 
specification of the problem.

In order to build the evolving graphs sequence suitable for the proof of our impossibility 
result, we need the following technical lemma.

\begin{lemma} \label{lemma_modification_direction_bis}
 Let $\mathcal{A}$ be a perpetual exploration algorithm in connected-over-time
 ring of size $4$ or more using 2 robots. Any execution of $\mathcal{A}$ satisfies:
 For any time $t$ and any robot state $s$, if, at time $t$,
 the robots have not explored the whole ring, have not formed a tower, and each
 robot has only visited at most two adjacent nodes, then there exists $t' \geq t$
 such that a robot located on a node $u$, on state $s$ at time $t$, and 
 satisfying $OneEdge(u, t, t')$ leaves $u$ at time $t'$.
\end{lemma}

\begin{proof}
 Consider an algorithm $\mathcal{A}$ that deterministically solves the 
 perpetual exploration problem for connected-over-time rings of 
 size $4$ or more using two robots.
 Let $\mathcal{G} = \{G_{0}=(V,E_0), G_{1}=(V,E_1), \ldots\}$ be a connected-over-time ring (of
 size $4$ or more). Let $\varepsilon$ be an execution of $\mathcal{A}$ by two robots
 $r_{1}$ and $r_{2}$ on $\mathcal{G}$.
 
 By contradiction, assume that there exists a time $t$ and a state $s$ such that 
 $(i)$ the exploration of the whole ring has not been done yet; $(ii)$ from time
 $0$ to time $t$ none of the robots have formed a tower; $(iii)$ at time $t$ each
 robot has only visited at most two adjacent nodes of $\mathcal{G}$; and $(iv)$ at 
 time $t$ one of the robot (without lost of generality, 
 $r_{1}$) is in a state $s$ such that, for any $ t' \geq t$, if $r_{1}$ is on a
 node $u$ of $\mathcal{G}$ satisfying $OneEdge(u, t, t')$, then it does not 
 leave $u$ at time $t'$.
 
 Let $\mathcal{R}$ be the set of nodes visited by $r_{1}$ from time
 $0$ to time $t$. Note that, at time $t$, as each robot has only
 visited at most two adjacent nodes, then 
 $1 \leq \textpipe \mathcal{R} \textpipe \leq 2$. Let $i$ (resp. $f$) be the
 node in $\mathcal{G}$ where $r_{1}$ is located at time $0$ 
 (resp. $t$). If $\textpipe \mathcal{R} \textpipe = 2$, let $a$ be the
 node of $\mathcal{R}$ such that $a \neq i$, otherwise (\ie 
 $\textpipe \mathcal{R} \textpipe = 1$) let $a = i$. By assumption, either 
 $f = i$ or $f$ is an adjacent node of $i$ and in this later case $a = f$. 
 
 We construct a connected-over-time ring $\mathcal{G}' = \{G_{0}', G_{1}', \ldots\}$
 (with $G'_i=(V',E'_i)$ for any $i\in\mathbb{N}$) 
 such that the underlying graph of $\mathcal{G}'$ contains 8 nodes
in the following way. Let $i_{1}'$ be an arbitrary node of $\mathcal{G}'$. Let us
construct nodes $i_{2}'$, $a_{1}'$, $a_{2}'$, $f_{1}'$, and $f_{2}'$ of $\mathcal{G}'$
in function of $i_{1}'$ and of nodes $i$, $a$, and $f$ of $\mathcal{G}$ as explained
by Figure~\ref{impossibility_ter}. Note that this construction ensures that $f_{1}'$ 
and $f_{2}'$ are adjacent in $\mathcal{G}'$ in any case.

 We denote by $r(k)$ (resp. $l(k)$) the adjacent edge in the clockwise 
 (resp. counter clockwise) direction of a node $k$.
 For any $j \in \{0,\ldots, t-1\}$, let $E'_j$ be the set $E_{\mathcal{G}'}$ 
 with the following set of additional constraints\footnote{Note that the construction
 of $i_{1}'$, $i_{2}'$, $a_{1}'$, $a_{2}'$, $f_{1}'$, and $f_{2}'$ ensures us that there
 is no contradiction between these constraints in all cases.}:
\[\left\{
\begin{array}{ll}
r(i_{1}') \in E_{j}' \text{ and } l(i_{2}') \in E_{j}' & \text{iff } r(i) \in E_{j}\\
l(i_{1}') \in E_{j}' \text{ and } r(i_{2}') \in E_{j}' & \text{iff } l(i) \in E_{j}\\
r(a_{1}') \in E_{j}' \text{ and } l(a_{2}') \in E_{j}' & \text{iff } r(a) \in E_{j}\\
l(a_{1}') \in E_{j}' \text{ and } r(a_{2}') \in E_{j}' & \text{iff } l(a) \in E_{j}\\
\end{array}
\right.\] 
For any $j \geq t$, let $E'_j$ be the set $E_{\mathcal{G}'}\setminus\{(f_{1}',f_{2}')\}$.

 Now, we consider the execution $\varepsilon'$ of $\mathcal{A}$ on $\mathcal{G}'$ starting
 from the configuration where $r_{1}$ (resp. $r_{2}$) is on node $i_{1}'$
 (resp. on node $i_{2}'$) such that the two robots have opposite
 chirality and that $r_{1}$ have the same chirality as in $\varepsilon$.
 The execution $\varepsilon'$ satisfies the following set of claims.
 
\begin{figure*}
	\begin{center}
		\includegraphics[scale=0.7]{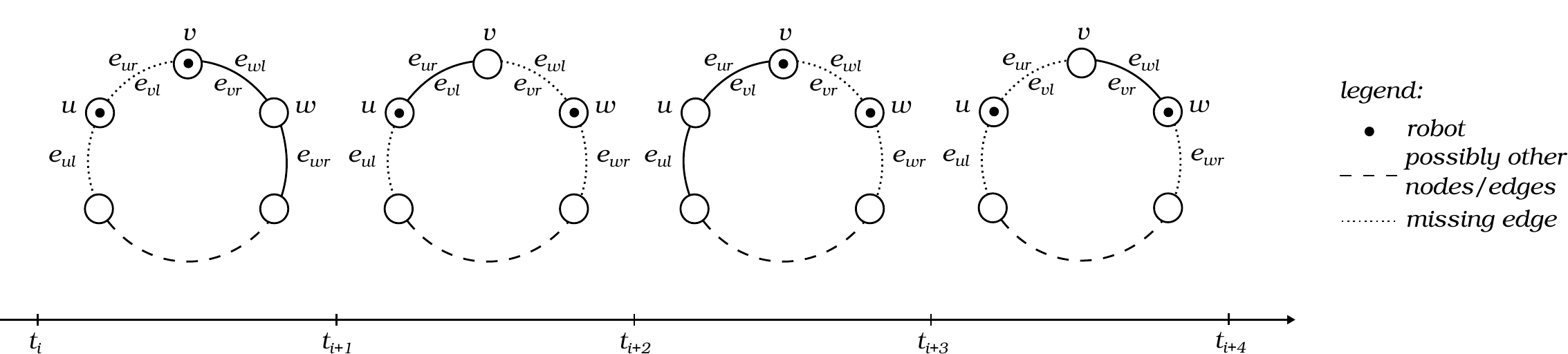}
	\end{center}
	\caption{Construction of $\mathcal{G}_{i + 1}$,  $\mathcal{G}_{i + 2}$, 
 $\mathcal{G}_{i + 3}$, and $\mathcal{G}_{i + 4}$ in proof of Theorem \ref{no_perpetual_exploration_two_robots}.} \label{impossibility_bis}
\end{figure*}

\noindent\textbf{Claim 1:} Until time $t$, $r_{1}$ and 
  $r_{2}$ execute the same actions in a symmetrical way in 
  $\varepsilon'$.
 
  \noindent Consider that, during the Look phase of time $j$, the two robots 
  have the same view in $\varepsilon'$. The two robots have not the same 
  chirality and $\mathcal{A}$ is deterministic, then, during the Move phase of
  time $j$, they are executing the same action in a symmetrical way
  (either not move or move in opposite directions). This implies that, at time 
  $j + 1$, $r_{1}$ and $r_{2}$ have again the same state. 
  
  \noindent There are only two robots executing $\mathcal{A}$ on $\mathcal{G}'$. Hence, if a
  tower is formed, it is composed of $r_{1}$ and $r_{2}$. If from time
  $0$ to time $t$, the robots are executing the same actions in a 
  symmetrical way, then, by construction of $\mathcal{G}'$ and by the way we 
  initially placed $r_{1}$ and $r_{2}$ on $\varepsilon'$, the two robots see
  the same local environment at each instant time in $\{0,\ldots, t\}$.
  
  \noindent At time $0$, by construction of $\mathcal{G}'$ and by the way we 
  placed $r_{1}$ and $r_{2}$ on $\varepsilon'$, the two robots have the same 
  view. 
  
  \noindent By recurrence and using the arguments of the two first paragraphs, we 
  conclude that, from time $0$ to time $t$, $r_{1}$ and $r_{2}$ execute
  the same actions in a symmetrical way in $\varepsilon'$. 
 
\noindent\textbf{Claim 2:} Until time ${t}$, $r_{1}$ and 
  $r_{2}$ never form a tower in $\varepsilon'$.
 
  \noindent By construction of $\varepsilon'$, the two robots are initially 
  at an odd distance. By Claim 1, at a time $0 <j + 1 < t$,
  the two robots are either at the same distance, at a distance increased of
  2, or at a distance decreased of 2 with respect to their distance at time $j$.
  Moreover, since $\mathcal{G}'$ possesses an even number of edges, this implies
  that, until time $t$, the robots are always at an odd distance from each other.

\noindent\textbf{Claim 3:} Until time ${t}$, $r_{1}$ executes in
  $\varepsilon'$ the same sequence of actions than in $\varepsilon$.
  
  \noindent Consider that, during the Look phase of time $j$, $r_{1}$ has the 
  same view in $\varepsilon$ and in $\varepsilon'$. As $\mathcal{A}$ is 
  deterministic, then, during the Move phase of time $j$, $r_{1}$ 
  executes the same action (either not move, or move in the same direction) in
  $\varepsilon$ and in $\varepsilon'$. This implies that, during the Look phase 
  of time $j + 1$, $r_{1}$ possesses the same state in $\varepsilon$ 
  and in $\varepsilon'$. 
 
  \noindent By assumption, until time $t$, there is no tower in $\varepsilon$.
  By Claim 2, there is no tower in $\varepsilon'$ until time $t$.
  Hence, in the case where $r_{1}$ executes 
  the same actions in $\varepsilon$ and in $\varepsilon'$ from time $0$ to time $t$, 
  $r_{1}$ sees the same local environment in $\varepsilon$ and in $\varepsilon'$ 
  until time $t$ (by construction of $\mathcal{G}'$ and the initial location of
  $r_{1}$ in $\varepsilon'$).

  \noindent At time $0$, $r_{1}$ has the same view in $\varepsilon$ and in
  $\varepsilon'$ (by construction of $\mathcal{G}'$ and the initial location of
  $r_{1}$ in $\varepsilon'$).
  
  \noindent By recurrence and using the arguments of the two first paragraphs, we 
  conclude that, from time $0$ to time $t$, $r_{1}$ executes the same
  actions in $\varepsilon$ and in $\varepsilon'$.
  
\noindent\textbf{Claim 4:} At time ${t}$, $r_{1}$ and $r_{2}$
  are on two adjacent nodes in $\varepsilon'$ and are both in state $s$.
  
  \noindent By Claims 1 and 3 and  by construction of $\mathcal{G'}$, 
  we know that at time $t$, $r_{1}$ is on node $f_{1}'$ while $r_{2}$ 
  is on node $f_{2}'$. These nodes are adjacent by construction of $\mathcal{G'}$.
 
 \noindent By Claim 1, as $r_{1}$ and $r_{2}$ have opposite chirality, they have 
 the same state at time $t$ in $\varepsilon'$. By Claim 3, $r_{1}$ is in the same state 
 at time $t$ in $\varepsilon$ and in $\varepsilon'$. Since $r_1$ is in state $s$ at 
 time $t$ in $\varepsilon$ by assumption, we have the claim.

 By construction of $\mathcal{G}'$, $f_{1}'$ (resp.
 $f_{2}'$) satisfies the property $OneEdge(f_{1}', t, +\infty)$ (resp.
 $OneEdge(f_{2}',$ $ t, +\infty)$). Then, by assumption, $r_{1}$ 
 (resp. $r_{2}$) does not leave node $f_{1}'$ (resp. $f_{2}'$) 
 after time $t$. As $\mathcal{G}'$ counts 8 nodes, we obtain a contradiction with 
 the fact that $\mathcal{A}$ is a deterministic algorithm solving the perpetual 
 exploration problem for connected-over-time rings using two robots.
\end{proof}

\begin{theorem} \label{no_perpetual_exploration_two_robots}
 There exists no deterministic algorithm satisfying the perpetual exploration
 specification on the class of connected-over-time rings of size $4$ or more
 with two fully synchronous robots.
\end{theorem}

\begin{proof}
 By contradiction, assume that there exists a deterministic algorithm $\mathcal{A}$ 
 satisfying the perpetual exploration specification on any connected-over-time ring
 of size $4$ or more with two robots $r_{1}$ and $r_{2}$.

 Consider the connected-over-time graph $\mathcal{G}$ whose underlying graph 
 $U_{\mathcal{G}}$ is a ring of size strictly greater than 3 such that all 
 the edges of $U_{\mathcal{G}}$ are present at each time.
 
 Consider three nodes $u$, $v$ and $w$ of $\mathcal{G}$, such that node $v$ is 
 the adjacent node of $u$ in the clockwise direction, and $w$ is the adjacent 
 node of $v$ in the clockwise direction.
 We denote respectively $e_{ur}$ and $e_{ul}$ the clockwise and counter clockwise
 adjacent edges of $u$, $e_{vr}$ and $e_{vl}$ the clockwise and counter 
 clockwise adjacent edges of $v$, and $e_{wr}$ and $e_{wl}$ the clockwise and
 counter clockwise adjacent edges of $w$. Note that $e_{ur} = e_{vl}$ and 
 $e_{vr} = e_{wl}$.
 
 Let $\varepsilon$ be the execution of $\mathcal{A}$ on $\mathcal{G}$ starting
 from the configuration where $r_{1}$ (resp. $r_{2}$) is located on node $u$ 
 (resp. $v$).
 
 We construct a sequence of connected-over-time graphs 
 ($\mathcal{G}_{n}$)$_{n \in \mathbb{N}}$ such that
 $\mathcal{G}_{0} = \mathcal{G}$ and for any $i \geq 0$, $\mathcal{G}_{i}$ is 
 defined as follows (denote by $\varepsilon_i$ the execution of $\mathcal{A}$ on 
 $\mathcal{G}_{i}$ starting from the same configuration as $\varepsilon$). 
 We define inductively $\mathcal{G}_{i + 1}$,  $\mathcal{G}_{i + 2}$, 
 $\mathcal{G}_{i + 3}$, and $\mathcal{G}_{i + 4}$ using Items 1-8 above 
 (see also Figure~\ref{impossibility_bis}) under the assumption that:
 $(i)$ $\mathcal{G}_{i}$ exists for a given $i \in \mathbb{N}$ multiple of 4; 
 $(ii)$ $\mathcal{G}_{i}$ is a connected-over-time ring;
 $(iii)$ there exists a time $t_{i}$ such that each robot has only visited at
 most two adjacent nodes among $\{u, v, w\}$ in $\varepsilon_i$;
 $(iv)$ before time $t_{i}$, the two robots never form a tower in $\varepsilon_i$; and
 $(v)$ at time $t_{i}$, $r_{1}$ (resp. $r_{2}$) is located on node $u$ (resp. $v$). 
 
 \begin{enumerate}[leftmargin=0cm,itemindent=.5cm,labelwidth=\itemindent,labelsep=0cm,align=left]
  \item Due to assumptions $(ii)$ to $(v)$, Lemma~
  \ref{lemma_modification_direction} implies that there exists a time 
  $t_{i}' \geq t_{i}$ such that $r_{2}$ leaves $v$ at time $t_{i}'$ if $r_{2}$ is
  located on node $v$ at time $t_{i}$ and $v$ satisfies $OneEdge(v, t_{i}, t_{i}')$. 

  We then define $\mathcal{G}_{i + 1}$ such that
  $U_{\mathcal{G}_{i + 1}} = U_{\mathcal{G}_{i}}$ and $\mathcal{G}_{i + 1} = \mathcal{G}_{i} 
  \backslash \{(e_{ul}, \{t_{i}, \ldots, t_{i}'\}),$ 
  $(e_{vl}, \{t_{i}, \ldots, t_{i}'\})\}$.

  Note that $\mathcal{G}_{i}$ and $\mathcal{G}_{i + 1}$ are indistinguishable for robots 
  before time $t_{i}$. This implies that, at time $t_{i}$, $r_{1}$ (resp. $r_{2}$) is 
  on node $u$ (resp. $v$) in $\varepsilon_{i + 1}$. 
  By construction of $t_{i}'$, $r_{2}$ leaves $v$ at time
  $t_{i}'$ in $\varepsilon_{i + 1}$. Since, at time $t_{i}'$, among the adjacent edges
  of $v$, only $e_{vr}$ is present in
  $\mathcal{G}_{i + 1}$, $r_{2}$ crosses this edge at this time in $\varepsilon_{i + 1}$. 
  Hence, at time $t_{i}' + 1$, $r_{2}$ is on node $w$ in $\varepsilon_{i + 1}$. Note that
  none of the adjacent edges of $r_{1}$ are present between time $t_{i}$ and
  time $t_{i}'$ in $\mathcal{G}_{i}$. That implies that, at time $t_{i}' + 1$, $r_{1}$ is
  still on node $u$ in $\varepsilon_{i + 1}$. Moreover, this construction ensures us that
  assumptions $(iii)$ and $(iv)$ are satisfied in $\varepsilon_{i + 1}$ until time 
  $t_{i}' + 1$. Finally, $\mathcal{G}_{i + 1}$ is a connected-over-time ring (since it is 
  indistinguishable from $\mathcal{G}$ after $t_{i}' + 1$) and hence satisfies 
  assumption $(ii)$.

  \item Let $t_{i + 1} = t_{i}' + 1$.
  
  \item Using similar arguments as in Item 1, we prove that there
  exists a time $t_{i + 1}'$ such that $r_{1}$ leaves $u$ at time 
  $t_{i + 1}'$ if $r_{1}$ is on node $u$ at time $t_{i+1}$ and $u$ satisfies
  $OneEdge(u, t_{i + 1}, t_{i + 1}')$. We define $\mathcal{G}_{i + 2}$ such that 
  $U_{\mathcal{G}_{i + 2}} = U_{\mathcal{G}_{i + 1}}$ and
  $\mathcal{G}_{i + 2} = \mathcal{G}_{i + 1} 
  \backslash \{(e_{ul}, \{t_{i + 1}, \ldots, t_{i + 1}'\}),$ 
  $(e_{wl}, \{t_{i + 1}, \ldots, t_{i + 1}'\}),$ 
  $(e_{wr}, \{t_{i + 1}, \ldots, t_{i + 1}'\})\}$. 

  That implies that, at time $t_{i + 1}' + 1$, $r_{1}$ (resp. $r_{2}$) is on node $v$ 
  (resp. $w$) in $\varepsilon_{i + 2}$ and that assumptions $(ii)$, $(iii)$, and $(iv)$ are 
  satisfied in $\varepsilon_{i + 2}$ until time $t_{i+1}' + 1$.
  
  \item Let $t_{i + 2} = t_{i + 1}' + 1$.
  
  \item Using similar arguments as in Item 1, we prove that there
  exists a time $t_{i + 2}'$ such that $r_{1}$ leaves $v$ at time 
  $t_{i + 2}'$ if $r_{1}$ is on node $v$ at time $t_{i+2}$ and $v$ satisfies 
  $OneEdge(v, t_{i + 2}, t_{i + 2}')$. We define $\mathcal{G}_{i + 3}$ such that 
  $U_{\mathcal{G}_{i + 3}} = U_{\mathcal{G}_{i + 2}}$ and such that 
  $\mathcal{G}_{i + 3} = \mathcal{G}_{i + 2} 
  \backslash \{(e_{wl}, \{t_{i + 2}, \ldots, t_{i + 2}'\}),$
  $(e_{wr}, \{t_{i + 2}, \ldots, t_{i + 2}'\})\}$.

  That implies that, at time $t_{i + 2}' + 1$, $r_{1}$ (resp. $r_{2}$) is on node $u$ 
  (resp. $w$) in $\varepsilon_{i + 3}$ and that assumptions $(ii)$, $(iii)$, and $(iv)$ are 
  satisfied in $\varepsilon_{i + 3}$ until time $t_{i+2}' + 1$.

  \item Let $t_{i + 3} = t_{i + 2}' + 1$.
  
  \item Using similar arguments as in Item 1, we prove that there
  exists a time $t_{i + 3}'$ such that $r_{2}$ leaves $w$ at time 
  $t_{i + 3}'$ if $r_{2}$ is on node $w$ at time $t_{i+3}$ and $w$ satisfies 
  $OneEdge(w, t_{i + 3}, t_{i + 3}')$. We define $\mathcal{G}_{i + 4}$ such that 
  $U_{\mathcal{G}_{i + 4}} = U_{\mathcal{G}_{i + 3}}$ and such that 
  $\mathcal{G}_{i + 4} = \mathcal{G}_{i + 3} 
  \backslash \{(e_{ul}, \{t_{i + 3}, \ldots, t_{i + 3}'\}),$ 
  $(e_{ur}, \{t_{i + 3}, \ldots, t_{i + 3}'\}),$
  $(e_{wr}, \{t_{i + 3}, \ldots, t_{i + 3}'\})\}$. 

  That implies that, at time $t_{i + 3}' + 1$, $r_{1}$ (resp. $r_{2}$) is on node $u$ 
  (resp. $v$) in $\varepsilon_{i + 4}$ and that assumptions $(ii)$, $(iii)$, and $(iv)$ are 
  satisfied in $\varepsilon_{i + 4}$ until time $t_{i+3}' + 1$.
  
  \item Let $t_{i + 4} = t_{i + 3}' + 1$.
 \end{enumerate}
 
 
 Note that $\mathcal{G}_{0}$ trivially satisfies assumptions $(i)$ to $(v)$ for $t_0=0$ 
 (since $\varepsilon_0=\varepsilon$ by construction). Also, given a $\mathcal{G}_{i}$
 with $i \in \mathbb{N}$ multiple of 4, $\mathcal{G}_{i+4}$ exists and we proved
 that it satisfies assumptions $(ii)$ to $(v)$. 
 In other words, ($\mathcal{G}_{n}$)$_{n \in \mathbb{N}}$ is well-defined.


 We define the evolving graph $\mathcal{G}_{\omega}$ such that
 $U_{\mathcal{G}_{\omega}} = U_{\mathcal{G}_{0}}$ and
\[
\begin{array}{r@{}c@{}l}
\mathcal{G}_{\omega} = \mathcal{G}_{0} \backslash \{ & (e_{ul}, & \{t_{4i}, \ldots, t_{4i}'\} \cup\{t_{4i + 1}, \ldots t_{4i + 1}'\}\cup\{t_{4i + 3}, \ldots, t_{4i + 3}'\}),\\
 & (e_{vl}, & \{t_{4i}, \ldots, t_{4i}'\}\cup\{t_{4i + 3}, \ldots, t_{4i + 3}'\}),\\
 & (e_{wl}, & \{t_{4i + 1}, \ldots, t_{4i + 1}'\}\cup\{t_{4i + 2}, \ldots, t_{4i + 2}'\}),\\
 & (e_{wr}, & \{t_{4i + 1},\ldots, t_{4i + 1}'\} \cup\{t_{4i + 2},\ldots, t_{4i + 2}'\}\cup\{t_{4i + 3},\ldots, t_{4i + 3}\})\textpipe i \in \mathbb{N}\}\\
\end{array}
\]

 Note that, for any edge of $\mathcal{G}_{\omega}$, the intervals of times where this
 edge is absent (if any) are finite and disjoint. This edge is so infinitely often present in
 $\mathcal{G}_{\omega}$. Therefore, $\mathcal{G}_{\omega}$ is a connected-over-time ring.
 
 For any $i\in\mathbb{N}$, $\mathcal{G}_{i}$ and $\mathcal{G}_{\omega}$ have a common prefix
 until time $t_{i}'$. As the sequence $(t_{n}$)$_{n \in \mathbb{N}}$ is increasing 
 by construction, this implies that the sequence ($\mathcal{G}_{n}$)$_{n \in \mathbb{N}}$
 converges to $\mathcal{G}_{\omega}$.

 Applying the theorem of \cite{BDKP16}, we obtain that, until time $t_{i}'$, the 
 execution of $\mathcal{A}$ on $\mathcal{G}_{\omega}$ is identical to the one on 
 $\mathcal{G}_{i}$. This implies that, executing $\mathcal{A}$ on 
 $\mathcal{G}_{\omega}$ (of size strictly greater than 3), 
 $r_{1}$ and $r_{2}$ only visit the nodes $u$, $v$, and $w$.
 This is contradictory with the fact that $\mathcal{A}$ satisfies the 
 perpetual exploration specification on connected over time rings of size 
 strictly greater than 3 using two robots. 
\end{proof}

\subsection{Connected-over-time Rings of Size $3$}

In this section, we present \PEFRR, a deterministic algorithm solving the
perpetual exploration on connected-over-time rings of size $3$ with two robots.

This algorithm works as follows. Each robot disposes only of its $dir$ variable.
If at a time $t$, a robot is isolated on a node with only one adjacent edge,
then it points to this edge. Otherwise (\ie none of the adjacent edge is present, both adjacent 
edges are present, or the other robot is present on the same node), the robot keeps its 
current direction.

\begin{theorem}\label{th:algo2robots}
 \PEFRR~is a perpetual exploration algorithm for the class of connected-over-time
 rings of 3 nodes using 2 fully synchronous robots.
\end{theorem}

\begin{proof}
Consider any execution of \PEFRR~on any connected-over-time
 ring of size 3 with 2 robots. By the connected-over-time 
assumption, each node has at least one adjacent edge infinitely often present. 
This implies that any tower is broken in finite time (as robots meet only when they 
consider opposite directions and move as soon as
it is possible). Two cases are now possible. 

\noindent\textbf{Case 1:} There exists infinitely often a tower in the execution.

\noindent Note that, if a tower is formed at a time $t$, then the three
nodes have been visited between time $t - 1$ and time $t$. Then, the 
three nodes are infinitely often visited by a robot in this case.

\noindent\textbf{Case 2:} There exists a time $t$ after which the robots are always isolated.

\noindent By contradiction, assume that there exists a time $t'$ such  that a node $u$
 is never visited after $t'$. As the ring has 3
nodes, that implies that, after time $max\{t,t'\}$, either the robots are always switching
their position or they stay on their respective nodes. 

\noindent In the first case, during the Look phase of each time greater than 
$max\{t,t'\}$, the respective 
variables $dir$ of the two robots contain the direction leading to $u$ (since 
it previously move in this direction). As at least one 
of the adjacent edges of $u$ is infinitely often present, a robot crosses it
in a finite time, that is contradictory with the fact that $u$ is not visited after $t'$.

\noindent The second case implies that both adjacent edges to the location of both robots
are always absent after time $t$ (since a robot moves as soon as it is possible), 
that is contradictory with the connected-over-time assumption.

In both cases, \PEFRR~satisfies the perpetual exploration specification.
\end{proof}

   \section{With One Robot}\label{sec:1robot}

 \begin{figure*}
 	\begin{center}
 		\includegraphics[scale=0.7]{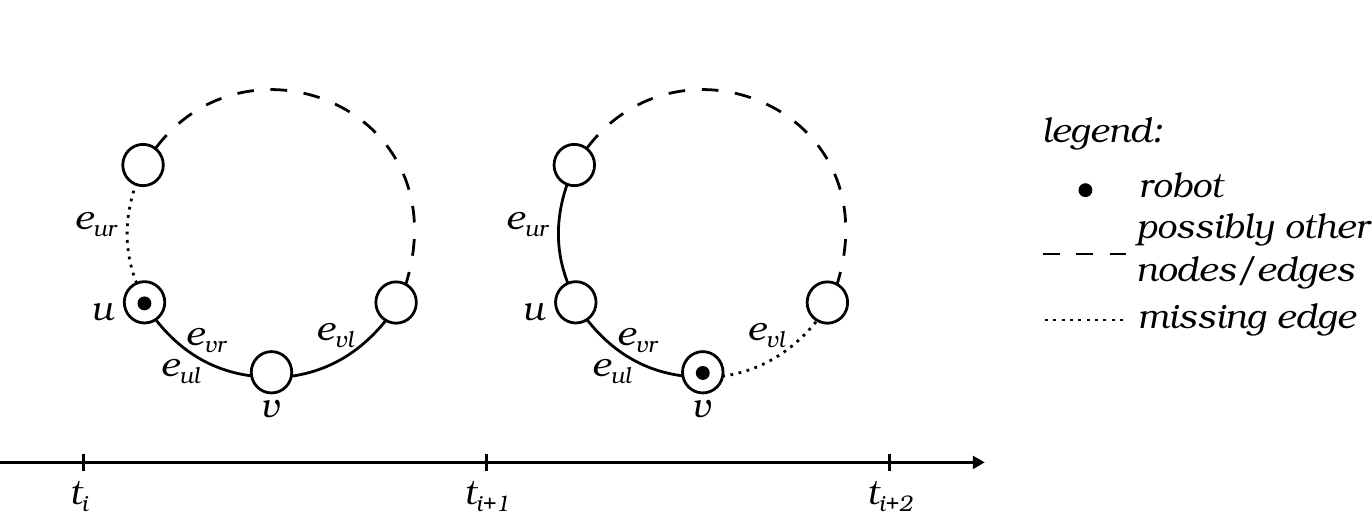}
 	\end{center}
 	\caption{Construction of $\mathcal{G}_{i + 1}$ and $\mathcal{G}_{i + 2}$
    in proof of Theorem \ref{no_perpetual_exploration_one_robot}.} \label{impossibility}
 \end{figure*}

This section leads a similar study than the one of Section~\ref{sec:2robots} but in the case of
the perpetual exploration of rings of any size with a single robot. 
Again, we first prove a negative result since Theorem~
\ref{no_perpetual_exploration_one_robot} states that a single
robot is not able to perpetually explore connected-over-time rings of size 
strictly greater than 2. We then provide \PEFR~(see Theorem \ref{th:algo1robot}),
an algorithm using a single robot that solves the perpetual exploration
on connected-over-time rings of size 2.

\subsection{Connected-over-time Rings of Size $3$ and More} \label{impossibility_1_robot}

Similarly to the previous section, the proof of our impossibility result presented 
in Theorem \ref{no_perpetual_exploration_one_robot} is based on the construction
of an adequate sequence of evolving graphs and the application of the generic framework 
proposed in \cite{BDKP16}.

In order to build the evolving graphs sequence suitable for the proof of our impossibility 
result, we need the following technical lemma.

\begin{lemma} \label{lemma_modification_direction}
 Let $\mathcal{A}$ be a perpetual exploration algorithm in connected-over-time
 ring of size $3$ or more using one robot. Any execution of $\mathcal{A}$ satisfies:
 For any time $t$ and any robot state $s$, there exists a 
 time $t' \geq t$ such that a robot located on a node $u$, on state $s$ at time $t$, and 
 satisfying $OneEdge(u, t, t')$ leaves $u$ at time $t'$.
\end{lemma}

\begin{proof}
 Consider an algorithm $\mathcal{A}$ that solves deterministically the 
 perpetual exploration problem for connected-over-time rings  of size
 $3$ or more using a single robot.
 Let $\mathcal{G} = \{G_{0}=(V,E_0), G_{1}=(V,E_1), \ldots\}$ be a 
 connected-over-time ring (of size $3$ or more). Let $\varepsilon$ be an 
 execution of $\mathcal{A}$ on $\mathcal{G}$ by a robot $r$.
 
 By contradiction, assume that it exists a time $t$ and a state $s$
 such that, for any $t' \geq t$, a robot $r$ located on a node $u$ of 
 $\mathcal{G}$ and in state $s$ at time $t$ with $u$ satisfying $OneEdge(u, t, t')$
 does not leave $u$ at time $t'$.
 
 Let $e$ be an arbitrary adjacent edge to $u$. Let us define the connected-over-time 
 ring $\mathcal{G}' = \{G_{0}'=(V,E'_0), G_{1}'=(V,E'_1), \ldots\}$ such that:
 \[\left\{\begin{array}{ll}
 E_{i}' = E_{i} & \text{if } i < t\\
 E_{i}' = E_{\mathcal{G}}\setminus\{e\} & \text{if } i \geq t\\
 \end{array}\right.
 \]
 Let $\varepsilon'$ be the execution of $\mathcal{A}$ on $\mathcal{G}'$ starting 
 from the same configuration than $\varepsilon$. 

 As $\mathcal{A}$ is a deterministic algorithm, $r$ is in the state $s$ and is
 located on node $u$  at time $t$ in $\varepsilon'$ by construction of $\mathcal{G}'$.
 Note that the node $u$ satisfies the property $OneEdge(u, t, +\infty)$ in $\mathcal{G}'$.
 
 Then, by assumption, $r$ does not leave $u$ in $\varepsilon'$ after time $t$.
 This implies that, after time $t$, only $u$ is visited in $\varepsilon'$.
 As $\mathcal{G}'$ counts $3$ or more nodes, we obtain a contradiction with 
 the fact that $\mathcal{A}$ is a deterministic algorithm solving the perpetual 
 exploration problem for connected-over-time rings using a single robot.
\end{proof}

\begin{theorem} \label{no_perpetual_exploration_one_robot}
 There exists no deterministic algorithm satisfying the perpetual exploration
 specification on the class of connected-over-time rings of size $3$ or more
 with a single fully synchronous robot.
\end{theorem}

\begin{proof}
 By contradiction, assume that there exists a deterministic algorithm $\mathcal{A}$ 
 satisfying the perpetual exploration specification on any connected-over-time ring
 of size $3$ or more with a single robot $r$.

 Consider the connected-over-time graph $\mathcal{G}$ whose underlying graph 
 $U_{\mathcal{G}}$ is a ring of size strictly greater than 2 such that all 
 the edges of $U_{\mathcal{G}}$ are present at each time.
 Consider any node $u$ of $\mathcal{G}$ and denote respectively by $e_{ur}$ 
 and $e_{ul}$ the clockwise and counter clockwise adjacent edges of $u$.

 Let $\varepsilon$ be the execution of $\mathcal{A}$ on $\mathcal{G}$ starting
 from the configuration where $r$ is located on node $u$.

 We construct a sequence of connected-over-time graphs 
 ($\mathcal{G}_{n}$)$_{n \in \mathbb{N}}$ such that
 $\mathcal{G}_{0} = \mathcal{G}$ and for any $i \geq 0$, $\mathcal{G}_{i}$ is 
 defined as follows (denote by $\varepsilon_i$ the execution of $\mathcal{A}$ on 
 $\mathcal{G}_{i}$ starting from the same configuration as $\varepsilon$).
 We define inductively $\mathcal{G}_{i + 1}$ and $\mathcal{G}_{i + 2}$
 using Items 1-4 above 
 (see also Figure~\ref{impossibility}) under the assumption that:
 $(i)$ $\mathcal{G}_{i}$ exists for a given $i \in \mathbb{N}$ even;
 $(ii)$ $\mathcal{G}_{i}$ is a connected-over-time ring; and 
 $(iii)$ there exists a time $t_{i}$ such that $r$ is located on node $u$ at time $t$ in $\varepsilon_i$. 

 \begin{enumerate}[leftmargin=0cm,itemindent=.5cm,labelwidth=\itemindent,labelsep=0cm,align=left]
  \item Due to assumptions $(ii)$ and $(iii)$, 
  Lemma~\ref{lemma_modification_direction} implies that there exists a time
  $t_{i}' \geq t_{i}$ such that $r$ leaves $u$ at time $t_{i}'$ if it is located on
  node $u$ at time $t_{i}$ and $u$ satisfies $OneEdge(u, t_{i}, t_{i}')$.

  We then define $\mathcal{G}_{i + 1}$ such that
  $U_{\mathcal{G}_{i + 1}} = U_{\mathcal{G}_{i}}$ and 
  $\mathcal{G}_{i + 1} = \mathcal{G}_{i}\backslash \{(e_{ur}, \{t_{i}, \ldots, t_{i}'\})\}$.

  Note that $\mathcal{G}_{i}$ and $\mathcal{G}_{i + 1}$ are indistinguishable for $r$ 
  before time $t_{i}$. This implies that, at time $t_{i}$, $r$ is located on node 
  $u$ in $\varepsilon_{i + 1}$.
  By construction of $t_{i}'$, $r$ leaves $u$ at time $t_{i}'$ in $\varepsilon_{i + 1}$.
  Since, at time $t_{i}'$, among the adjacent edges of $u$, only $e_{ul}$ is present in
  $\mathcal{G}_{i + 1}$, $r$ crosses this edge at this time in $\varepsilon_{i + 1}$. 
  Then, at time $t_{i}' + 1$, $r$ is located on node $v$ (the node adjacent to $u$ 
  in the counter clockwise direction) in $\varepsilon_{i + 1}$.
  Finally, $\mathcal{G}_{i + 1}$ is a connected-over-time ring (since it is 
  indistinguishable from $\mathcal{G}$ after $t_{i}' + 1$) and hence satisfies 
  assumption $(ii)$.
  Denote respectively by $e_{vr}$ and $e_{vl}$ the clockwise and counter clockwise
  adjacent edges of $v$. We have $e_{ul} = e_{vr}$.
  
  \item Let $t_{i + 1} = t_{i}' + 1$. 

  \item Using similar arguments as in Item 1, we prove that there
  exists a time $t_{i + 1}'$ such that $r$ leaves $v$ at time $t_{i + 1}'$ if
  $r$ is located on node $v$ at time $t_{i+1}$ and $v$ satisfies 
  $OneEdge(v, t_{i + 1}, t_{i + 1}')$. We define $\mathcal{G}_{i + 2}$ such that 
  $U_{\mathcal{G}_{i + 2}} = U_{\mathcal{G}_{i + 1}}$ and
  $\mathcal{G}_{i + 2} = \mathcal{G}_{i + 1} 
  \backslash \{(e_{vl}, \{t_{i + 1}, \ldots t_{i + 1}'\})\}$. 

  That implies that, at time $t_{i + 1}' + 1$, $r$ is on node $u$ 
  in $\varepsilon_{i + 2}$ and that assumptions $(ii)$ and $(iii)$ are 
  satisfied in $\varepsilon_{i + 2}$ until time $t_{i+1}' + 1$.
  
  \item Let $t_{i + 2} = t_{i + 1}' + 1$.
 \end{enumerate}
 
 Note that $\mathcal{G}_{0}$ trivially satisfies assumptions $(i)$ to $(iii)$ for $t_0=0$ 
 (since $\varepsilon_0=\varepsilon$ by construction). Also, given a $\mathcal{G}_{i}$
 with $i \in \mathbb{N}$ even, $\mathcal{G}_{i+2}$ exists and we proved
 that it satisfies assumptions $(ii)$ and $(iii)$. 
 In other words, ($\mathcal{G}_{n}$)$_{n \in \mathbb{N}}$ is well-defined.

 We define the evolving graph $\mathcal{G}_{\omega}$ such that
 $U_{\mathcal{G}_{\omega}} = U_{\mathcal{G}_{0}}$ and
\[\begin{array}{r@{}l}
 \mathcal{G}_{\omega} = \mathcal{G}_{0} \backslash \{ & (e_{ur}, \{t_{2i}, \ldots t_{2i}'\}), 
 (e_{vl}, \{t_{2i+1}, \ldots, t_{2i+1}'\}) \textpipe i \in \mathbb{N}\}\\
\end{array}
\]
 Note that, for any edge of $\mathcal{G}_{\omega}$, the intervals of times where this
 edge is absent (if any) are finite and disjoint. This edge is so infinitely often present in
 $\mathcal{G}_{\omega}$. Therefore, $\mathcal{G}_{\omega}$ is a connected-over-time ring.

 For any $i\in\mathbb{N}$, $\mathcal{G}_{i}$ and $\mathcal{G}_{\omega}$ have a common prefix
 until time $t_{i}'$. As the sequence $(t_{n}$)$_{n \in \mathbb{N}}$ is increasing 
 by construction, this implies that the sequence ($\mathcal{G}_{n}$)$_{n \in \mathbb{N}}$
 converges to $\mathcal{G}_{\omega}$.

 Applying the theorem of \cite{BDKP16}, we obtain that, until time $t_{i}'$, the 
 execution of $\mathcal{A}$ on $\mathcal{G}_{\omega}$ is identical to the one on 
 $\mathcal{G}_{i}$. This implies that, executing $\mathcal{A}$ on 
 $\mathcal{G}_{\omega}$ (of size strictly greater than 2), 
 $r$ only visits the nodes $u$ and $v$.
 This is contradictory with the fact that $\mathcal{A}$ satisfies the 
 perpetual exploration specification on connected over time rings of size 
 strictly greater than 2 using one robot. 
\end{proof}

\subsection{Connected-over-time Rings of Size $2$}

In this section, we present \PEFR, a deterministic algorithm solving the
perpetual exploration on connected-over-time rings of size $2$ with a single robot.

Note that a ring of size $2$ can be defined in two different ways. If we consider
that the graph must remain simple, such a ring is reduced to a 2-node chain (\ie only one 
bidirectional edge links the two nodes). Otherwise (\ie the graph may be not simple),
the two nodes are linked by two bidirectional edges. In both cases, 
the following algorithm, \PEFR, trivially works 
as follow: As soon as at least one adjacent edge to the current node of the
robot is present, its variable $dir$ points arbitrarily to one of these edges.

\begin{theorem}\label{th:algo1robot}
 \PEFR~is a perpetual exploration algorithm for the class of connected-over-time
 rings of 2 nodes using a single fully synchronous robot.
\end{theorem}


   \section{Conclusion}\label{sec:conclu}

We analyzed the computability of the perpetual 
exploration problem on highly dynamic rings. We proved that three (resp., two) 
robots with very few capacities are necessary to solve the perpetual exploration problem on
connected-over-time rings that include strictly more than three (resp., two) nodes. 
For the completeness of our work, we provided three algorithms: 
One for a single robot evolving in a 2-node ring, one for two robots exploring three nodes, and 
one for three or more robots moving among at least four nodes. 
These three algorithms allow to show that the necessary number of robots is also 
sufficient to solve the problem.

%


   \newpage
   \bibliographystyle{plain}
   \bibliography{biblio}
\end{document}